\documentclass[11pt]{article}
\usepackage{graphicx}
\usepackage{amsmath,amssymb,amsthm,mathtools}
\usepackage{paralist}
\usepackage{bm}
\usepackage{xspace}
\usepackage{url}
\usepackage{fullpage, prettyref}
\usepackage{boxedminipage}
\usepackage{wrapfig}
\usepackage{ifthen}
\usepackage{color}
\usepackage{xcolor}
\usepackage{framed}
\usepackage{algorithmic,algorithm}
\usepackage[pagebackref,letterpaper=true,colorlinks=true,pdfpagemode=none,urlcolor=blue,linkcolor=blue,citecolor=violet,pdfstartview=FitH]{hyperref}
\usepackage{fullpage}
\usepackage{thmtools}
\usepackage{thm-restate}

\newtheorem{theorem}{Theorem}[section]

\newtheorem{lemma}[theorem]{Lemma}
\newtheorem{claim}[theorem]{Claim}
\newtheorem{corollary}[theorem]{Corollary}
\newtheorem{definition}[theorem]{Definition}

\newtheorem{observation}[theorem]{Observation}

\newcommand{\ignore}[1]{}

\newcommand{\cD}{\mathcal{D}}

\newcommand{\cL}{{\cal L}}

\newcommand{\cP}{\mathcal{P}}

\newcommand{\cU}{{\cal U}}

\newcommand{\B}{\mathbf{B}}
\newcommand{\be}{{\bf e}}

\newcommand{\R}{\mathbb R}

\newcommand{\eps}{\varepsilon}

\newcommand{\bD}{\boldsymbol{D}}

\newcommand{\bR}{\boldsymbol{R}}

\newcommand{\Exp}{\EX}

\newcommand{\EX}{\hbox{\bf E}}

\newcommand{\Sec}[1]{\hyperref[sec:#1]{\S\ref*{sec:#1}}} 
\newcommand{\Eqn}[1]{\hyperref[eq:#1]{(\ref*{eq:#1})}} 
\newcommand{\Fig}[1]{\hyperref[fig:#1]{Fig.\,\ref*{fig:#1}}} 
\newcommand{\Tab}[1]{\hyperref[tab:#1]{Tab.\,\ref*{tab:#1}}} 
\newcommand{\Thm}[1]{\hyperref[thm:#1]{Theorem\,\ref*{thm:#1}}} 
\newcommand{\Fact}[1]{\hyperref[fact:#1]{Fact\,\ref*{fact:#1}}} 
\newcommand{\Lem}[1]{\hyperref[lem:#1]{Lemma\,\ref*{lem:#1}}} 
\newcommand{\Prop}[1]{\hyperref[prop:#1]{Prop.~\ref*{prop:#1}}} 
\newcommand{\Cor}[1]{\hyperref[cor:#1]{Corollary~\ref*{cor:#1}}} 
\newcommand{\Conj}[1]{\hyperref[conj:#1]{Conjecture~\ref*{conj:#1}}} 
\newcommand{\Def}[1]{\hyperref[def:#1]{Definition~\ref*{def:#1}}} 
\newcommand{\Alg}[1]{\hyperref[alg:#1]{Alg.~\ref*{alg:#1}}} 
\newcommand{\Ex}[1]{\hyperref[ex:#1]{Ex.~\ref*{ex:#1}}} 
\newcommand{\Clm}[1]{\hyperref[clm:#1]{Claim~\ref*{clm:#1}}} 

\def\dist{{\sf dist}}

\def\fii{{f^{(i)}}}
\def\fext{f_{\textrm{\tt ext}}}
\def\ef{g_{\textrm{\tt ext}}}
\def\efg{g'_{\textrm{\tt ext}}}

\def\fia{f^{(i)}_{|\a}}
\def\a{{\mathbf a}}

\def\max{{\sf max}}

\def\pdiext{d_{\textrm{\tt ext}}}
\def\pdiext{\pdi_{\textrm{\tt ext}}}

\def\pdi{{\mathfrak m}}
\def\cPdd{{\cP}}

\newcommand{\set}[1]{\{#1\}}

\newcommand{\frest}[2]{{#1}_{|#2}}
\newcommand{\slice}{S}

\def\VG{\mathcal{G}_{\mathsf{viol}}}

\colorlet{shadecolor}{blue!05}
\begin{document}
\title{$L_p$-Testers for Bounded Derivative Properties on Product Distributions}

\author{
Kashyap Dixit\thanks{Pennsylvania State University, {\tt kashyap@cse.psu.edu}, supported in part by NSF Grant CCF-0964655 and CCF-1320814}
}

\date{}
\maketitle

\begin{abstract}
We consider the problem of $L_p$-testing of class of bounded derivative properties over hypergrid domain with points distributed according to some product distribution. 
This class includes monotonicity, the Lipschitz property, $(\alpha,\beta)$-generalized Lipschitz and many more properties. Previous results for $L_p$ testing on $[n]^d$ for this class 
were known for monotonicity and $c$-Lipschitz properties over uniformly distributed domains. \medskip

Our results imply testers that give the same upper bound for arbitrary product distributions as the hitherto known 
testers, which use uniformly randomly chosen samples from $[n]^d$, for monotonicity and Lipschitz testing. Also, our testers are \emph{optimal} for a large class of bounded derivative properties, 
that includes  $(\alpha, \beta)$-generalized Lipschitz property, over uniform distributions. Infact, each edge in $[n]^d$ is allowed to have it's own left and right Lipschitz constants.
 The time complexity is \emph{same} for arbitrary product distributions.

\end{abstract}
\thispagestyle{empty}
\newpage
\setcounter{page}{1}
\pagenumbering{arabic}

\section{Introduction}
The field of \emph{property testing}~\cite{RS96,GGR98} deals with the following question:
can one decide whether a function $f$ has a certain property or not,
while accessing it only on a tiny fraction of its domain?
To address this informational bottleneck, the decision problem is relaxed to distinguish  
functions having the property from functions which are `far' from having the property.
This needs a measure of \emph{distance}, $\dist(f,g) \in (0,1)$ between functions. A function is said to be
$\eps$-far from the property if $\dist(f,g)>\eps$ whenever $g$ satisfies the property.


The notion of distance is central to property testing.
The conventional definition of distance 
is the Hamming distance over the domain with respect to some distribution $\cD$, that is,  $\dist(f,g) := \Pr_{x \sim \cD} [f(x) \neq g(x)]$.
So, if a function $f$ is not $\eps$-far, then there exists a function $g$ satisfying the property,
and samples drawn uniformly at random
cannot distinguish $f$ and $g$ with probability more than $\eps$. Most of the testers that have been designed in the past give high probability guarantees for when the samples are drawn from the uniform distribution and distance to the property is the hamming distance. Two notable recent detours from this approach are~\cite{BeRaYa14} and \cite{paper}
The former gives the first systematic study of the property testing with the notion of farness being $L_p$-distance from the property. The later gives the \emph{optimal} testers for properties when queries are made over domain points sampled from some product distribution. We unify the two settings and get \emph{better} results for the \emph{$L_p$-property testing over product distributions}.

\Def{lp} formally defines the notion of $L_p$-distance over general distributions. This definition from \cite{BeRaYa14} generalizes the notion of distance given in the seminal work of 
Goldreich, Goldwasser, and Ron~\cite{GGR98}. \cite{BeRaYa14} showed wide applications of $L_p$-testing in the fields of learning, approximation theory, noise suppression etc. The authors in~\cite{paper} 
give optimal testers for a class of properties called bounded derivative properties. Owing to the generality of bounded derivative properties refered to as $\cP$, their work subsumes almost all the works done
 in the field of monotonicity and the Lipschitz $L_0$-property testing over past decade.

In this paper we study the  problem of {\em bounded derivative property testing} of real valued functions  $f:[n]^d\mapsto [a,b]$, with respect to 
a {\em product distribution} $\cD := \prod_{i\leq d}\cD_i$ and $L_p$-distance. First, we show how to extend the testers of~\cite{BeRaYa14} to the product distribution setting. Then we generalize the Lipschitz tester of~\cite{BeRaYa14}
to a much broader class of bounded derivative properties that includes $(\alpha,\beta)$-Lipschitz property, where $\alpha$ and $\beta$ are the left and right Lipschitz constants. 
To be precise, each edge is allowed to have it's own personal set of left and right Lipschitz constants.We note that out results match the adaptive lower bound for $c$-Lipschitz testing 
given in~\cite{BeRaYa14}, which is a special case. Therefore, the upper bound is \emph{optimal} when the bounding family is a set of constant valued functions.

Roughly, a function $f$ is in $\cP$  iff the discrete analogue of it's first derivative is bounded. \Def{bound} formally describes $\cP$.
The problem is to distinguish monotone functions from those which are $\eps$-far with respect to $L_p$-distance defined by $\cD$ over $[n]^d$. Well studied properties like monotonicity, 
the Lipschitz property, $(\alpha,\beta)$-generalized Lipschitz property and many more properties fall in the purview of bounded derivative properties. 
The bounds may be set in such a way that the function is required $c$-Lipschitz in first dimension, $(\alpha,\beta)$-Lipschitz in second dimension, $(\alpha',\beta')$-Lipschitz in third dimension and so on. 

\subsection{Preliminaries}\label{sec:prelims}

%
%
The notion of $L_p$ distance is defined in~\cite{BeRaYa14} which resembles closely to the following. \smallskip
\noindent
\begin{definition}[\bf $L_p$- distance]\label{def:lp} Let $f$ be a real valued function over finite domain $D$. For $p\ge 1$, the $L_p$-norm of $f$ is $(\sum_{x\in D}|f(x)|^p)^{\frac{1}{p}}$ ($p=0,1,2$). 
Let $||f_0||$ be the number of non-zero values of $f$. A property $\cP$ is the set of functions over $D$. For real valued functions $f:D\mapsto [0,r]$, we define the following distance measure.
$$d_{\cP}=\frac{1}{r}\cdot\inf_{g\in P}(\Exp[|f-g|^p])^{\min\{\frac{1}{p},1\}}$$
\end{definition}

In the following, $f:[n]^d \mapsto \bR$ is a fixed function
and $\cD = \prod_{i=1}^d \cD_i$ is a product distribution over $[n]^d$. 
For a subset $X \subseteq [n]^d$, we use $\mu_\cD(X)$ 
to denote the probability mass of the subset $X$, and  $\mu_{\cD_i}$ to denote the marginal along the $i$th dimension. Therefore, for any point $x\in [n]^d$, we have $\mu_\cD(x) = \prod_{i=1}^d \mu_{\cD_i}(x_i)$. 
We let $\cU$ denote the uniform distribution; observe that $\mu_\cU(X) =|X|/n^d$.

A line along dimension $i$, or simply an $i$-line, in $[n]^d$ is a collection of $n$ points which have all but their $i$th coordinates same. 
Given a dimension $i$, we let $\cD_{-i}$ denote the distribution $\prod_{j\neq i}\cD_j$.
Observe that $\cD_{-i}$ is product distribution on $i$-lines.
For any line $\ell$, the restriction of $f$ to the line $\ell$ is denoted by $\frest{f}{\ell}$. Note
that $\frest{f}{\ell}$ has domain $[n]$.

We will use the following theorem from~\cite{BeRaYa14} to prove our result for all $L_p$-norms. Note that since we obtain the upper bounds for product distributions by reducing them to uniform distribution, 
\Thm{bry1} applies directly to our setting as well.
\begin{theorem}[\cite{BeRaYa14}]\label{thm:bry1}
For every property $\cP$ over any domain and all $\varepsilon\in (0,1)$
$$1.~ Q_0(\cP,\varepsilon)\ge  Q_1(\cP,\varepsilon);~~~2.~Q_0(\cP,\varepsilon)\ge  Q_1(\cP,\varepsilon); ~~~3. ~Q_1(\cP,\varepsilon^2)\ge Q_2(\cP, \varepsilon)$$
\noindent
Moreover, if $\cP$ is a property of boolean functions then $Q_0(\cP,\varepsilon)=Q_1(\cP,\varepsilon)=Q_0(\cP,\sqrt{\varepsilon})$.
\end{theorem}
\noindent
Now we formally define the bounding function family associated with some bounded derivative property $\cP$.
\begin{definition}[\cite{paper}] \label{def:bound} An ordered set $\B$ of $2d$ functions 
$l_1, u_1, l_2, u_2, \ldots, l_d, u_d: [n-1] \mapsto \R$
is called a \emph{bounding family} if for all $r \in [d]$ and $y \in [n-1]$, $l_r(y) < u_r(y)$.
Let $\B$ be a bounding family of functions.
The property of being \emph{$\B$-derivative bounded}, denoted as $\cP(\B)$, is the set of functions
$f:[n]^d \mapsto \R$ such that: for all $r \in [d]$ and $x \in [n]^d$,
\begin{equation}
\label{eq:defnbnd}
l_r(x_r) \leq \partial_r f(x) \leq u_r(x_r).
\end{equation}
\end{definition}

We define a quasimetric depending on $\B$ denoted by $\pdi(x,y)$.
\begin{definition}[\cite{paper}] \label{def:dist} Given bounding family $\B$, construct the weighted directed hypergrid $[n]^d$, where 
all adjacent pairs are connected by two edges in opposite directions. The weight of $(x+\be_r,x)$ is $u_r(x_r)$ and the weight
of $(x,x+\be_r)$ is $-l_r(x_r)$. $\pdi(x,y)$ is the shortest path weight from $x$ to $y$.
\end{definition}\noindent
Note that $\pdi$ is asymmetric, can take negative values, and $\pdi(x,y) = 0$ does not necessarily imply $x=y$. 
For these reasons, it is ceases to be a metric, although we will refer to it simply as a metric in the remainder of the paper abusing its closeness to a metric due to the properties stated in \Lem{dist-prop}.
It has been shown in \cite{paper} that
\begin{equation}
\label{eq:supergeneralLip}
\pdi(x,y) :=  \sum_{r:x_r > y_r} \sum_{t = y_r}^{x_r-1}\! u_r(t)  -  \sum_{r:x_r < y_r}\sum_{t = x_r}^{y_r-1}\!l_r(t)
\end{equation}
If a function $f\in \cP(\B)$, then applying \Eqn{defnbnd} on every edge of the path described above (the upper bound when we decrement a coordinate and the lower bound when we increment a coordinate), we get $f(x) - f(y) \leq \pdi(x,y)$ for any pair $(x,y)$. Conversely, if $\forall x,y, f(x) - f(y) \leq \pdi(x,y)$, then considering neighboring pairs gives $f\in \cP(\B)$.
This argument is encapsulated in the following lemma.
\begin{lemma} \label{lem:dist} $f \in \cP(\B)$ iff $~\forall x, y \in [n]^d$, $f(x) - f(y) \leq \pdi(x,y)$.
\end{lemma}
\noindent
When $\cP(\B)$ is monotonicity, $\pdi(x,y) = 0$ if $x \prec y$ and $\infty$ otherwise. For the $c$-Lipschitz property,
$\pdi(x,y) = c\|x-y\|_1$.
The following properties of $\pdi()$ are proved in \cite{paper}.
\begin{lemma}[\cite{paper}]\label{lem:dist-prop} $\pdi(x,y)$ satisfies the following properties.
\begin{asparaenum}
\item {\em (Triangle Inequality.)} For any $x,y,z$, $\pdi(x,z) \leq \pdi(x,y) + \pdi(y,z)$.
\item {\em (Linearity.)} If $x,y,z$ are such that for every $1\leq r\leq d$, either $x_r \leq y_r \leq z_r$ or $x_r \geq y_r \geq z_r$, then 
$\pdi(x,z) = \pdi(x,y) + \pdi(y,z)$.
\item {\em (Projection.)} Fix any dimension $r$. Let $x,y$ be two points with $x_r = y_r$. Let $x'$ and $y'$
be the projection of $x, y$ onto some other $r$-hyperplane. That is, 
$x'_r = y'_r$, and $x'_j = x_j$, $y'_j = y_j$  for $j\neq r$. Then, $\pdi(x,y) = \pdi(x',y')$ and $\pdi(x,x') = \pdi(y,y')$.
\end{asparaenum}
\end{lemma}
\begin{definition}[{\bf Violation Graph}]
The violation graph of a function $f$ with respect to property $\cPdd$, denoted as $\VG(f,\cPdd)$ (also $G_f$ in some of the subsequent sections), has $[n]^d$ 
as vertices, and edge $(x,y)$ if it is a violation to $\cPdd$, that is either $f(x) - f(y) > \pdi(x,y)$ or $f(y) - f(x) > \pdi(y,x)$.
\end{definition}

\subsection{Our results}\label{sec:res}\label{sec:results}

Our main result is the $L_p$ tester for bounded-derivative properties over the domains with arbitrary product distributions defined over them.
 As noted in \cite{paper}, this class includes properties like monotonicity~(\cite{DGLRRS99}), the Lipschitz property~(\cite{JR11}), $(\alpha,\beta)$-generalized Lipschitz 
property~(\cite{ChSe13}) and many more (depending on the bounding function family).
In particular, our tester implies the same upper bound of~\cite{BeRaYa14} for monotonicity and $c$-Lipschitz properties over product distribution. 

\begin{theorem}\label{thm:mono}
Consider functions $f:[n]^d \mapsto [a,b]$ equipped with a product distribution $\prod_{i=1}^d \cD_i$ and proximity parameter $\varepsilon\in (0,1)$.  
The time complexity of $L_p$ ($p=\{1,2\}$) testing of monotonicity of $f$ is $O(\frac{d}{\varepsilon^p}\log (\frac{d}{\varepsilon^p}))$. 
\end{theorem}

We also obtain first non-trivial bounds for other bounded derivative properties where the bounding function family is a set of constant valued functions. 

\begin{theorem} \label{thm:main-known} Consider functions 
$f:[n]^d \mapsto [a,b]$ equipped with a product distribution $\prod_{i=1}^d \cD_i$ and bounding family of functions $\B=\{l_i:[n-1]\mapsto S_l,u_i:[n-1]\mapsto S_u\}_{i=1}^{d}$
 corresponding to the property $\cP(\B)$. There is a tester for $\cP(\B)$ (sometimes refered as $\cP$) with running time $O(\frac{d}{\varepsilon^p})$. 
\end{theorem}
In \Thm{main-known}, the sets $S_l$ and $S_u$ are finite sets ($|S_u|\le n-1, |S_l|\le n-1$) of constants that are independent of $n$ and $b-a$. Note that the running time of tester in \Thm{main-known} matches 
the known non-adaptive lower bound for $c$-Lipschitz property (refer \cite{BeRaYa14}) which is a special case of $\cP$. Therefore, this is upper bound is \emph{optimal} for uniformly distributed domains.
In rest of the paper, we will be working with $L_1$ distance only. All the results can be extended to $L_2$ distance using \Thm{bry1}.

{\bf Scope of this work.} \Thm{main-known} covers a large set of properties including $(\alpha, \beta)$-Lipschitz property where $\alpha$ and $\beta$ are left and right Lipschitz constants. Infact it allows 
each edge in the hypergrid to have it's own personal left and right Lipschitz constants. Still, this is a strict subset of bounded-derivative properties as defined in~\cite{paper}. In particular, monotonicity testing is out of the 
scope of \Thm{main-known} because right derivative bound is not constant. The lower bound in~\cite{BeRaYa14} shows that such a bound cannot be achieved by any non-adaptive tester.

\subsection{Related work}\label{sec:rel}
The area of property testing has florished in the last fifteen years. A lot of work has been done for testing the bounded derivative properties like monotonicity~\cite{EKK+00, GGLRS00,DGLRRS99,LR01,FLNRRS02,AC04,E04,HK04,PRR04,ACCL07,BRW05,BGJRW09,BCG+10,BBM11,ChSe13,ChSe13-2,BeRaYa14,BlJh+13,paper} and the Lipschitz property~\cite{JR11, AJMS12, ChSe13, DiJh+13, BlJh+13,paper}.
With the exception of~\cite{HalevyK07, HK04,AC04, DiJh+13,paper}, all the previous works are in the uniform distribution setting. The work in~\cite{paper} shows tight non-adaptive upper bounds with adaptive lower bounds for their
 bounded derivative property tester over product distributions, thus resolving the open question posed by~\cite{AC04} about specific question of monotonicity testing over product distributions.

Goldreich et. al.~\cite{GGR98} had already posed the question of testing properties of functions over non-uniform distributions,
and obtain some results for dense graph properties. A serious study of the role of distributions
was undertaken by Halevy and Kushilevitz~\cite{HalevyK07,HK04,HalevyK05,HalevyK08}, who
formalized the concept of distribution-free testing. (Refer to Halevy's thesis~\cite{Hal-thesis} for a comprehensive study.)
Glasner and Servedio~\cite{GlasnerS09} and Dolev and Ron~\cite{DolevR11} give various upper and lower bounds for
distribution-free testers for various classes of functions  (not monotonicity) over $\set{0,1}^n$.


The field of $L_p$-testing is still relatively very little explored for $p=1,2$. The work by~\cite{FV13} shows gives an $L_1$-testers for submodularity. $L_1$-distance is widely used to study the properties of 
distributions~\cite{BFRSW13,BFRV11,CDVV14,DDSVV13,Valiant11,VV11}. Recently,~\cite{BeRaYa14} has given a systematic study of $L_p$-testing for various properties over uniformy distributed hypergrid domain $[n]^d$. 
They give many applications of $L_p$-property testing in various areas of computing.

\section{Reducing from arbitrary product distributions}\label{sec:uniftogen}
We reduce arbitrary product distributions to uniform distributions on what we call the bloated hypergrid.
Assume without loss of generality that all $\mu_{\cD_r}(j) = q_r(j)/N$, for some integers $q_r(j)$ and $N$.
Consider the $d$-dimensional $N$-hypergrid $[N]^d$. There is a natural many-to-one mapping from  $\Phi: [N]^d\mapsto [n]^d$ defined as follows. First fix a dimension $r$. 
Given an integer $1\leq t \leq N$, let $\phi_r(t)$ denote the index $\ell \in [1,n]$ such that 
$\sum_{j <\ell} q_r(j) < t \leq \sum_{j\leq \ell} q_r(j)$. That is, partition $[N]$ into $n$ contiguous segments of lengths $q_r(1),\ldots,q_r(n)$. Then $\phi_r(t)$ is the index of the segment where $t$ lies. The mapping $\Phi: [N]^d\mapsto [n]^d$ is defined as 
\[\Phi(x_1,x_2\ldots,x_d) = \left(\phi_1(x_1), \phi_2(x_2), \ldots, \phi_\pdi(x_d) \right).\]
We use $\Phi^{-1}$ to define the set of preimages, so $\Phi^{-1}$ maps a point in $[n]^d$ to a `cuboid' in $[N]^d$. 
Observe that for any $x\in [n]^d$,
\begin{equation}\label{eq:obs}
|\Phi^{-1}(x)| = N^d\prod_{r=1}^d \mu_{\cD_r}(x) = N^d\mu_\cD(x).
\end{equation}
\begin{claim}\label{clm:XtoZ}
For any set $X\subseteq [n]^d$, define $Z \subseteq [N]^d$ as $Z := \bigcup_{x\in X}\Phi^{-1}(x)$. Then $\mu_\cD(X) = \mu_\cU(Z)$.
\end{claim}
\begin{proof}
The set  $Z = \bigcup_{x\in X}\Phi^{-1}(x)$ is the union of all the preimages of $\Phi$ over the elements of $X$. Since preimages are disjoint, we get
$|Z| = \sum_{x\in X}|\Phi^{-1}(x)| = N^d\mu_\cD(X)$. Therefore, $\mu_\cU(Z) = \mu_\cD(X)$.
\end{proof}
Given $f:[n]^d \mapsto \R$, we define its extension $\fext:[N]^d \mapsto \R$: 
\begin{equation}
\label{eq:fext}\fext(x_1,\ldots,x_d) = f(\Phi(x_1,\ldots,x_d)). 
\end{equation}
Thus, $\fext$ is constant on the cuboids in the bloated hypergrid corresponding to a point in the original hypergrid. 
Define the following metric on $[N]^d$. 
\begin{equation}\label{eq:dext}
\textrm{For $x,y \in [N]^d$,} \quad \pdiext(x,y) =\pdi(\Phi(x),\Phi(y)) 
\end{equation}
The following statements establish the utility of the bloated hypergrid, and the proof of the dimension reduction of $f$ over $[n]^d$ w.r.t. $\cD$ follows easily from these and the proof for the uniform distribution.
\begin{lemma}[\cite{paper}] \label{lem:distconsistency}
If $\pdi$ satisfies the conditions of \Lem{dist-prop} over $[n]^d$, then so does
$\pdiext$ over $[N]^d$.
\end{lemma}

Let $g:[n]^d\mapsto \R$ be the function in $\cP$ that is closest to $f$ and $\ef:[N]^d\mapsto \R$ be the extension of $g$ over $[N]^d$. Also, let $\efg:[N]^d\mapsto \R$ be the function closest to $\fext$ in $\cP$. \Thm{bhg} shows that the reduction preserves the $L_p$-distance.

\begin{theorem}\label{thm:bhg}
$(\Exp[|f-g|^p)])^{\frac{1}{p}} =(\Exp[|\fext-\ef|^p)])^{\frac{1}{p}}=(\Exp[|\fext-\efg|^p])^{\frac{1}{p}}$.
\end{theorem}
\begin{proof}
Note that the first equality in the theorem because of the following.
\begin{eqnarray}\label{eq:1st}\Exp[|f-g|^p]=\sum_{x\in [n]^d}\mu_{\cD}(x)|f(x)-g(x)|^p=\sum_{x\in [n]^d}\sum_{v\in [N]^d:\Phi(v)=x}\frac{|\fext(v)-\ef(v)|^p}{N^d}=\Exp[|\fext-\ef|^p]\end{eqnarray}
\noindent
Consider two points $v_1,v_2\in [N]^d$ such that $\Phi(v_1)=\Phi(v_2)$. Since $\efg\in \cP$, we have $\efg(v_1)-\efg(v_2)\le \pdiext(v_1,v_2)=\pdi(\Phi(v_1),\pdi(\Phi(v_2)))=0$. Similarly, $\efg(v_2)-\efg(v_1)\le \pdi(\Phi(v_2),\Phi(v_1))=0$.
Therefore $\efg(v_1)=\efg(v_2)$. Therefore, for every $x\in [n]^d$, $\efg$ takes the same value for all points $\{v\in [N]^d:\Phi(v)=x\}$. 

Consider the function $g':[n]^d\mapsto \R$ defined as $g'(\Phi(v))=\efg(v)$. Note that by the chain of equalities similar to \Eqn{1st}, it can be shown that $\Exp[|f-g'|^p]=\Exp[|\fext-\efg|^p]$. Note that $\Exp[|f-g|^p]\le \Exp[|f-g'|^p]$ and $\Exp[|\fext-\efg|^p]\le \Exp[|\fext-\efg|^p]$.
This and \Eqn{1st} yield that $\Exp[|\fext-\ef|^p]=\Exp[|\fext-\efg|^p]$.
\end{proof}

In the subsequent sections, we will talk about the testers in the uniformly distributed hypergrid domain. All the results can be extended to the known product distribution using bloated hypergrid domain arguments.

\section{$L_1$-distance and maximum weight matching}
Let $f:\bD\mapsto \bR$ be a function on discrete domain $\cD$ with induced quasimetric $d$ induced by the bounded derivative property $\cP$.  A pair $(x,y)$ is not violated iff $f(x)-f(y)\le \pdi(x,y)$. The violation score of $(x,y)$, defined as $vs_f(x,y)=\max\{f(x)-f(y)-\pdi(x,y), f(y)-f(x)-\pdi(y,x)\}$. An edge is violated iff $vs_f(x,y)>0$. The violation graph $G_f=(\bD, E_f)$ of $f$ is a graph on $\bD$ such that $(x,y)\in E_f$ iff $vs_f(x,y)>0$. The weight of each edge $(x,y)\in E_f$ is $vs_f(x,y)$. The following lemma relates the maximum weight matching in $G_f$ to $L_1$ distance of $f$ from $\cP$. We note that \Lem{l1-mwm} generalizes Lemma 3.1 in~\cite{BeRaYa14} which was proved for Lipschitz property.  

We need the following observation about the violation score in further discussions.
\begin{claim}\label{clm:vs+}
If $f(x)-f(y)-\pdi(x,y)>0$ for some pair $(x,y)$, then $vs_f(x,y)=f(x)-f(y)-\pdi(x,y)$
\end{claim}
\begin{proof}
We have $f(x)-f(y)>\pdi(x,y)$ which implies that $f(y)-f(x)<-\pdi(x,y)<\pdi(y,x)$ from lemma[to be filled]. Therefore $vs_f(x,y)=\max\{f(x)-f(y)-\pdi(x,y), f(y)-f(x)-\pdi(y,x)\}=f(x)-f(y)-\pdi(x,y)$.
\end{proof}
 
\begin{lemma}\label{lem:l1-mwm}
Let $M$ be the maximum weight matching in $G_f$. Let $vs_f(M)=\sum_{(x,y)\in E_f}vs_f(x,y)$. Then $L_1(f,\cP)=vs_f(M)$
\end{lemma}

\begin{proof}
First we prove that $L_1(f,\cP)\ge vs_f(M)$. Let $g:\bD\mapsto \bR$ be the closest function to $f$ in $\cP$, that is, $L_1(f, \cP)=L_1(f,g)$. Consider a violated edge $(x,y)\in G_f$. W.l.o.g., assume that $vs_f(x,y)=f(x)-f(y)-\pdi(x,y)$. Note that $f(x)-f(y)> \pdi(x,y)$ and $g(x)-g(y)\le \pdi(x,y)$. Therefore we have $|f(x)-g(x)|_1+|f(y)-g(y)|_1\ge (f(x)-g(x))-(f(y)-g(y))\ge f(x)-f(y)-\pdi(x,y)=vs_f(x,y)$. Moreover, since $M$ is a matching, each vertex appears in at most one pair in $M$. Therefore $L_1(f,g)\ge \sum_{(x,y)\in G_f}vs_f(x,y)$.

Now we prove that $L_1(f,\cP)\le vs_f(M)$. The argument in this part is closely related to the proof of Lemma 3.1 in~\cite{BeRaYa14}. Consider the partition of vertex set $\bD=\{\bD_>,\bD_=,\bD_<\}$. where $\bD_{op}=\{x\in \bD| f(x)~op~g(x)\}$ where $op=\{>,=,<\}$. Consider the bipartite graph $B_f=(\{\bD_>\cup \bD_=\}\times \{\bD_<\cup \bD_=\}, E_B)$. The edges $E_B$ consist of pairs $(x,y)\in \{\bD_>\cup \bD_=\}\times \{\bD_<\cup \bD_=\}$ such that $g(x)-g(y)=\pdi(x,y)$ and $x,y$ map to different vertices in $G_f$.

Note that $f(x)\ge g(x)\ge g(y)\ge f(y)$. Since $f(x)-f(y)>g(x)-g(y)=\pdi(x,y)$, we have $vs_f(x,y)=f(x)-f(y)-\pdi(x,y)$ by \Clm{vs+}. Thus we have 
$$vs_f(x,y)= f(x)-f(y)-\pdi(x,y)=f(x)-f(y)-(g(x)-g(y))=|f(x)-g(x)|+|f(y)-g(y)|$$

\Lem{hall} shows that $B_f$ contains a matching $M'$ which matches every vertex $x\in \bD_<\cup \bD_>$. We say that $M'\ni x$ if $x$ is matched in $M'$. Thus, we infer that 
$$vs_f(M')=\sum_{(x,y)\in M'}vs_f(x,y)= \sum_{M'\ni x}|f(x)-g(x)|=\sum_{x\in \bD_<\cup \bD_>}|f(x)-g(x)|=L_1(f,\cP)$$

Now we show that there exists a matching $M\in G_f$ such that $vs_f(M)\ge vs_f(M')$. Consider an edge $(x,y)\in B_f$. If $x\in \bD_>$ and $y\in \bD_<$, then we add $(x,y)$ to $M$ in $B_f$. Consider a vertex $x\in \bD_=$ in $B_f$. Let the edges incident to $x$ in $B_f$ be $(y,x)$ and $(x,z)$. Note that by \Lem{hall}, both the edges exist. Also, by the definition of $B_f$, we have $f(y)>f(x)>f(z)$. Now we have 
$$vs_f(y,x)+vs_f(x,z)= (f(y)-g(y))+(g(x)-f(x))+(f(x)-g(x))+(g(z)-f(z))=|f(y)-g(y)|+|f(z)-g(z)|$$
Therefore $vs(y,z)=vs(y,x)+vs(x,z)$. Therefore $vs(M')=vs(M)$. 
\end{proof}

\begin{lemma}\label{lem:hall}
$B_f$ contains a matching $M'$ that matches every vertex in $\cD_<\cup \cD_>$
\end{lemma}

The proof of this lemma is very similar to the one given in~\cite{BeRaYa14}, so we skip the details here.

\section{Dimension reduction from the grid to the line}\label{sec:dimred}
A natural approach for bounded derivative property testing over hypergrids is to perform
\emph{dimension reduction}~\cite{DGLRRS99,HK04,AC04,BeRaYa14, paper} to the line:
Can one argue that if $f$ is far from $\cP$ on the hypergrid,
then some (or many) restrictions $\frest{f}{\ell}$ to lines will also be far from $\cP$?

Consider the expected distance of $\frest{f}{\ell}$, where $\ell$ is a random axis parallel chosen from $\cL(n,d)$ which denotes the set of all axis-parallel lines in $[n]^d$. 
Our dimension reduction lemma is the following.
\begin{lemma}\label{lem:grid-to-line-gen} (Dimension reduction)
$ \sum_{\ell \in \cL(n,d)}[L_1(\frest{f}{\ell},\cP)] \geq L_1(f,\cP)/2$.
\end{lemma}
We note here that the dimension reduction lemma with same guarantee has been proven in~\cite{BeRaYa14} for $c$-Lipschitz and monotonicity for boolean valued functions separately, but it is unclear to us how to extend it to the bounded derivative properties 
and  monotonicity of real valued functions, in particular when the bounding functions are different in each dimension.

%
%

\subsection{Dimension reduction and the alternating paths}\label{sec:cs}

We begin by proving \Lem{grid-to-line-gen} for the uniform distribution. This requires
some of the machinery of~\cite{ChSe13, paper}.
 Consider a maximum weight matching of minimum cardinality (henceforth called MWm)
$M$ of $G_f$. From \Lem{l1-mwm},  we get $vs_f(M) = \L_1(f,\cP)$.
An important theorem of~\cite{paper} related the size of $M$ to the number of special axis-aligned violated pairs.
This leads to an optimal property tester w.r.t the uniform distribution.

For a matching $M$ and coordinate $i$ we say that a pair $(x,y)\in M$ is an $i$-cross pair if $x_i\neq y_i$. 
%
%
%
%

\begin{theorem}[from~\cite{paper}]
\label{thm:no-viol}
Let $f$ be a function on $[n]^d$ and suppose there are no violations along the $i$-lines. That is, for any pair $(x,y)$ with $x_j = y_j$ for $j\neq i$ and $x_i < y_i$, we have $f(x)\leq f(y)$. Then there exists a MWM in $G_f$ with no $i$-cross pairs.
\end{theorem}

We define a sequence of matchings starting with the MWM $M$. The $i$th matching $M_i$ is also an MWM but it is not allowed any $j$-cross pairs for $1\leq j\leq i$. Our main claim (\Lem{mi}) proves that $L_1(\frest{f}{\ell}),\cP)$ is an upper bound on the drop in the violation scores of matchings, which gives the dimension reduction lemma.
This claim is proved using \Thm{no-viol} and uses a simple but useful structural statement about MWMs (\Clm{matchdiff}). We dive into the details now, starting with some definitions.
\begin{asparaitem}
	
\item Define $M_i$ to be a MWM in $G_f$
that does not contain any $j$-cross pairs for $1\leq j\leq i$.
Observe that $vs_f(M_0) = L_1(f,\cP)$ and $vs_f(M_d) = 0$. 
	\item {\em Hypergrid slices}: Given an $i$-dimensional vector $\a\in [n]^i$, let $\slice_\a := \{x\in [n]^d: x_j = \a_j, ~1\leq j\leq i\}$ be the {\em slice} of the hypergrid with respect to vector $\a$, or simply the $\a$-slice. Each slice is in itself an $[n]^{d-i}$ hypergrid.
The hypergrid $[n]^d$ is partitioned by the various $\a$-slices of the same dimension. That is, 
$[n]^d = \bigcup_{\a\in [n]^i} \slice_\a$, and $\slice_\a \cap \slice_{\a'} = \emptyset \textrm{ whenever }  \a\neq\a'$.
Given a function $f$ defined on the hypergrid, let $\frest{f}{\a}$ denote the function restricted to the slice 
$\slice_\a$.
\end{asparaitem}
\medskip
%
\begin{claim}\label{clm:matchdiff}
Let $f$ and $g$ be two real valued functions defined on a hypergrid. Let $M$ and $N$ be the MWMs w.r.t $f$ and $g$ on the respective violation graphs. Then, 
$|vs_f(M)-vs_f(N)| \leq L_1(f,g)$.
\end{claim}
\begin{proof}
Look at the symmetric difference of $M$ and $N$. This gives us alternating paths and cycles. Let $p_f\in \cP$ and $p_g\in \cP$ be two real valued functions on the hypergrid that are closest to $f$ and $g$, that is, $L_1(f,\cP)=L_1(f,p_f)$ and $L_1(g,\cP)=L_1(g,p_g)$. 

\begin{eqnarray*}
vs(M)-vs(N)&=&\sum_{x\in [n]^d}|f(x)-p_1(x)|-\sum_{x\in [n]^d}|g(x)-p_2(x)|\\
&=& \sum_{x\in [n]^d}(|f(x)-p_1(x)|-|g(x)-p_2(x)|)\\
&\le& \sum_{x\in [n]^d}(|f(x)-p_2(x)|-|g(x)-p_2(x)|)\\
&\le&\sum_{x\in [n]^d}(|f(x)-p_2(x)-(g(x)-p_2(x))|)\\
&=&\sum_{x\in [n]^d}(|f(x)-g(x)|)=L_1(f,g)\\
\end{eqnarray*}

The first inequality follows from the fact that $L_1(f,p_2)\ge L_1(f,p_1)$. The second inequality follows from the triangle inequality. Similarly, one can prove that $vs(N)-vs(M)\le L_1(f,g)$. Hence the claim.

\end{proof}
\noindent
\Lem{grid-to-line-gen} for the uniform distribution follows from the following lemma since $vs_f(M_0)= L_1(f, \cP)$. 


\begin{lemma}\label{lem:mi}
Let $\cL_i$ denote the set of lines that are axis-parallel to dimension $i$ for all $1\leq i\leq d$. Then we have 
$$2\cdot\sum_{\ell\in \cL_i} L_1(\frest{f}{\ell}, \cP)\geq \left(vs_f(M_{i-1}) -vs_f(M_{i})\right)$$.%
\end{lemma}

\begin{proof}


Since $M_{i-1}$ has no $j$-cross pairs for $j\leq i-1$, all pairs of $M_{i-1}$ have both endpoints in the same slice $\slice_\a$ for some $\a\in [n]^{i-1}$. Thus, $M_{i-1}$ partitions into sub-matchings in each $\slice_\a$. Let $M^\a_{i-1}$ be the pairs of $M_{i-1}$ with both end points in slice $\slice_\a$.
\begin{equation}\label{eq:1}
vs_f(M_{i-1}) = \sum_{\a\in [n]^{i-1}} vs_f(M^\a_{i-1})
\end{equation}
Similarly, $M^\a_i$ is defined, and infact since $M_i$ has no $i$-cross pairs either, we get for every $\a\in [n]^{i-1}$,
$vs_f(M^\a_{i}) = \sum_{j=1}^n vs_f(M^{(\a\circ j)}_{i})$.
where $(\a\circ j)$ is the $i$-dimensional vector obtained by concatenating $j$ at the end of $\a$.
The following is a consequence of the partition of $M_{i-1}$ and $M_i$ across the slices.
\begin{observation}
For any $\a\in [n]^{i-1}$, $M^\a_{i-1}$ is the MWM in $\slice_\a$ wrt $\frest{f}{\a}$. Furthermore, for any $j\in [n]$, $M^{(\a\circ j)}_i$ is the MWM in $\slice_{(\a\circ j)}$ wrt $\frest{f}{(\a\circ j)}$.
\end{observation}
Let $\fii$ be the closest function to $f$ with no violations along dimension $i$, that is, for any $x,y$ with $x_i<y_i$ and $x_j = y_j, j\neq i$, we have $\fii(x) < \fii(y)$.
Note that by definition, $L_1(f,\fii)=\sum_{\ell \in \cL_i}{L_1(\frest{f}{\ell},\cP)}$. 
 
Fix $\a\in [n]^{i-1}$ and focus on the $\a$-slice $\slice_\a$.  Note that  $\fii$ has no violations along the $i$-lines, 
neither does $\frest{f^{(i)}}{\a}$. Therefore, by \Thm{no-viol} there exists an MWM $N^\a$ in $\slice_\a$ w.r.t $\frest{f^{(i)}}{\a}$ that
has no $i$-cross pairs.
Therefore, $N^\a$ partitions as $N^\a = \bigcup_{j=1}^n N^{(\a\circ j)}$. Furthermore, each matching $N^{(\a\circ j)}$ is the MWM in $\slice_{(\a\circ j)}$ wrt $f^{(i)}_{|(\a\circ j)}$.

\noindent
Since $M^\a_{i-1}$ is a MWM wrt $\frest{f}{\a}$ and $N^\a$ is a MWM wrt $\fia$ in $\slice_\a$, \Clm{matchdiff} gives 
\begin{equation}\label{eq:3}
vs_f(N^\a) \geq vs_f(M^\a_{i-1}) - L_1(\frest{f}{\a},\frest{f^{(i)}}{\a})
\end{equation}
Since $M^{(\a\circ j)}_i$ is a MWM wrt $f_{|(\a\circ j)}$ and  $N^{(\a\circ j)}$ is a MWM  wrt $f^{(i)}_{|(\a\circ j)}$ in $\slice_{(\a\circ j)}$ , \Clm{matchdiff} gives us 
$vs_f(M^{(\a\circ j)}_i) \geq vs_f(N^{(\a\circ j)}) - L_1(\frest{f}{(\a\circ j)}, \frest{f^{(i)}}{(\a\circ j)})$. Summing over all $j$,
\begin{equation}\label{eq:4}
vs_f(M^\a_i) \geq vs_f(N^\a) - L_1(\frest{f}{\a},\frest{f^{(i)}}{\a})
\end{equation}
Adding \Eqn{3},\Eqn{4} over all $\a\in [n]^{i-1}$ gives 
$vs_f(M_i) \geq vs_f(M_{i-1}) - 2\sum_{\a\in [n]^{i-1}}L_1(\frest{f}{\a},\frest{f^{(i)}}{\a}) = |M_{i-1}| - 2\cdot L_1(f,f^{(i)})$.
 Adding over all $i\in [d]$ proves \Lem{grid-to-line-gen} since $ L_1(f,f^{(i)})=\sum_{\ell\in \cL_i}L_1(\frest{f}{\ell}, \cP)$.
\end{proof}

\section{Bounded derivative testing on a line}\label{sec:line}
Consider a function $f:[n]\mapsto [a,b]$, with bounding functions being $l:[n-1]\rightarrow S_l$ and $u:[n-1]\mapsto S_u$ such that $l(x)\le \partial f(x)\le u(x)$ for each $x\in [n-1]$. We note here that for monotonicity, the bounding 
family becomes $l_i=0$ and $u_i=r$ for all $i$. Given two functions $f:D\mapsto [0,r]$ and $g:D\mapsto [0,r]$, we define $dist(f,g)=\frac{||f-g||_1}{|D|\cdot r}$ .Therefore $dist(f,g)\in [0,1]$ and is \emph{scale-invariant}.

From now on, we will assume that $l(x)=-u(x)$. This assumption makes the analysis of the line tester much cleaner. The following reduction shows that this assumption is not without loss of generality. 

We define a function $g:[n]\mapsto [a',b']$ as $g(x)=f(x)+\sum_{v=x}^{n-1}\frac{u(v)+l(v)}{2}$. Also, the bounding family of $g$ is defined 
as $\{u',v'\}$ where $l'(v)=-\frac{(u(v)-l(v))}{2}$ and $u'(v)=\frac{(u(v)-l(v))}{2}$.  Note that since $u(v)>l(v)$, $l'(v)$ and $u'(v)$ are always non-zero. 
\Clm{negpos} shows that the violation score of each edge with respect to $g$ is same as with respect to $f$. 
\begin{claim}\label{clm:negpos}
Let $f$ and $g$ be defined as above. Then for any edge $(x,y)$ ($x<y$), $vs_f(x,y)=vs_g(x,y)$.
\end{claim}
\begin{proof}
We have 
\begin{eqnarray}\label{eq:ftog1}
g(x)-g(y)-\pdi_g(x,y) &=& f(x)-f(y)+\sum_{v=x}^{y-1}\frac{(u(v)+l(v))}{2}-\sum_{v=x}^{y-1}\frac{(u(v)-l(v))}{2}\nonumber\\&=&f(x)-f(y)+\sum_{x}^{y-1}l(v)=f(x)-f(y)-\pdi_f(x,y)
\end{eqnarray}
Similarly we have 
\begin{eqnarray}\label{eq:ftog1}
g(y)-g(x)-\pdi_g(y,x) &=& f(y)-f(x)-\sum_{v=x}^{y-1}\frac{(u(v)+l(v))}{2}-\sum_{v=x}^{y-1}\frac{(u(v)-l(v))}{2}\nonumber\\&=&f(y)-f(x)-\sum_{x}^{y-1}u(v)=f(y)-f(x)-\pdi_f(y,x)
\end{eqnarray}

Therefore we have 
$$vs_g(x,y)=\max\left\{f(x)-f(y)-\pdi_f(x,y),f(y)-f(x)-\pdi_f(y,x)\right\}=vs_f(x,y)$$
\end{proof}

\begin{corollary}\label{cor:mwm}
Maximum weight matchings in the violation graphs of $f$ and $g$ are identical.
\end{corollary}
Note that in the above reduction, the range size ($b'-a'$) of $g$ might be much larger than $b-a$, but this is not a problem since our final bound is independent of range size.

Suppose  we are given $f:[n-1]\mapsto [a,b] $  and the bounding functions $\{-u, u\}$ where $u:[n-1]\mapsto S_u$. Here $S_u$ is the set of positive constants that are independent of $n$ and $b-a$. Let $b-a=r$, $u_M=\max\{x\in [n]:u(x)\}$ and $u_m=\min\{x\in [n]:u(x)\}$

\begin{shaded}
{\bf Line Tester ($f$)}
\begin{asparaenum}
\itemsep0em 
\item Let $P_f$ be the set of pairs $\{(x,y)|x<y\}$ such that $y-x\le \frac{r}{u_m}$.\label{item:1}
\item Query an uniformly randomly selected pair from $P_f$.
\item if $vs_f(x,y)>0$ then {\bf reject}
\item {\bf accept}
\end{asparaenum}
\end{shaded}

Consider a violated pair $(x,y)$ ($x<y$) such that $f(x)-f(y)>\pdi_f(x,y)$. Now we have
$$r\ge f(x)-f(y)\ge \sum_{v=x}^{y-1}u(v)\ge u_m\cdot (y-x)$$
Therefore $y\le x+\frac{r}{u_m}$. Same bound can be obtained by considering the case when $f(y)-f(x)-\pdi_f(y,x)>0$. Therefore all the violated pairs lie in $P_f$
Note that total number of pairs in $P_f$ is at most $n\cdot\min \{n-1, \frac{r}{u_m}\}$.  We first give the lower bound on the number of violated pairs in $P_f$ denoted by $V(P_f)$. We get the lower bound using \Clm{lin} and \Clm{many}.
\begin{claim}\label{clm:lin}
Let $(x,y)$ be a pair violated by $f$ over $[n]$ and let $v=\left\lceil\frac{vs_{f}(x,y)}{2\cdot u_M}\right\rceil-1$. Then for all $z\in [x-v,y+v]\cap [n]$, $f$ violates at least one of the unordered pairs $(x,z)$ or $(y,z)$
\end{claim}
\begin{proof}
Note that the claim is true for $z\in [x,y]\cap [n]$ by linearity of quasimetric $\pdi_{f}$ induced by the bounding function family. Consider the case when $z\in[x-v,x]\cap [n]$. The case for $z\in [y,y+v]$ is analogous. Assume that $vs_f(x,y)= f(x)-f(y)-\pdi_{f}(x,y)$. The other case is symmetric. 
Assume for the contradiction that $(x,z)$ and $(y,z)$ are both not violated. Then $f(x)-f(z)\le \pdi_{f}(x,z)$ and $f(z)-f(y)\le \pdi_{f}(z,y)$. Adding both the inequalities we get $f(x)-f(y)\le \pdi_{f}(x,z)+\pdi_{f}(z,y)$. We also have $f(x)-f(y)=\pdi_{f}(x,y)+vs_{f}(x,y)$. Therefore we have
\begin{eqnarray*}
vs_{f}(x,y)+\pdi_{f}(x,y)&\le& \pdi_{f}(x,z)+\pdi_{f}(z,y)\\
vs_{f}(x,y)&\le& \pdi_{f}(x,z)+\pdi_{f}(z,y)-\pdi_{f}(x,y)\\
vs_{f}(x,y)&\le& \pdi_{f}(x,z)+\pdi_{f}(z,x) \text{ (since $\pdi_{f}(z,y)=\pdi_{f}(z,x)+\pdi_{f}(x,y)$)}
\end{eqnarray*}

Therefore we have $vs_{f}(x,y)\le  \pdi_{f}(x,z)+\pdi_{f}(z,x)=\sum_{\alpha=z}^{x-1}(u(\alpha)-l(\alpha))\le (2\cdot u_M)(x-z)$
This implies that $z\le x-\frac{vs_{f}(x,y)}{2\cdot u_M}$. Note that this is a contradiction since
 $$z\ge x-\left\lceil \frac{vs_{f}(x,y)}{2\cdot u_M}\right\rceil+1>x-\frac{vs_{f}(x,y)}{2\cdot u_M}$$ 
\end{proof}

\begin{claim}\label{clm:many}
Let $(x,y)$ be a pair violated by $f$ over $[n]$. Then $f$ violates at least $\min\{\frac{vs_{f}(x,y)}{2\cdot u_M}, n-1\}$ pairs.
\end{claim}
\begin{proof}
Let $v=\left\lceil \frac{vs_{f}(x,y)}{2\cdot u_M}\right\rceil-1$. If $x-v\ge 1 $ then $f$ violates either $(x,z)$ or $(y,z)$ for each $z\in \{x-v,\dots ,x-1\}$. Therefore there are $v+1\ge \left\lceil \frac{vs_{f}(x,y)}{2\cdot u_M}\right\rceil$ pairs including $(x,y)$ that are violated. 
Similarly, if $y+v\le n$, then one gets at least $v+1$ violated pairs $\{(u,z)\}\cup (x,y)$ for all $z$ in $\{y+1,\dots,y+v\}$ and $u\in \{x,y\}$. Finally, if $x-v<1$ and $y+v>n$ then $f'$ violates at least one of $(x,z)$ and $(y,z)$ for $z\in [n]\setminus\{x,y\}$. 
Thus there are at least $n-1$ violated pairs including $(x,y)$.
\end{proof}

Let $M$ be the maximum weight matching in $G_{f}$. Let $M_1$ be the set of pairs in $M$ with violation score at most $u_M\cdot (n-1)$ and $M_2=M\setminus M_1$. By \Clm{many}, 
each of the edges in $M_1$ should contribute at least $\frac{vs_{f}(M_1)}{2\cdot u_M}$ violated pairs.
 Let $g:[n]\mapsto [a,b]$ be the function in $\cP$ that is closest to $f$. Note that $vs_{f}(x,y)=|f(x)-g(x)|+|f(y)-g(y)|\le 2r$. Hence, the violation score of any pair can not exceed $2r$. 
Therefore there are at least $\frac{vs_{f}(M_2)}{2r}$ edges in $M_2$. By \Clm{many}, $M_2$ must 
contribute at least $\frac{vs_f(M_2)}{2r}\cdot (n-1)$ violated pairs. Since each violated pair is contributed by the edge of $M$ at most twice, the number of violated pairs in 
$M$ is at least
 $$\frac{1}{2}\left(\frac{vs_f(M_1)}{2\cdot u_M}+vs_f(M_2)\cdot \frac{n-1}{2r}\right)\ge \frac{vs_f(M)}{4}\cdot\min\left\{\frac{1}{u_M},\frac{n-1}{r}\right\}=\frac{\varepsilon}{4}\cdot n\cdot \min\left\{\frac{r}{u_M},n-1\right\} $$

The last equality follows from the fact that $\varepsilon=dist(f,\cP)=\frac{L_1(f,\cP)}{n\dot r}=\frac{vs_f(M)}{n\dot r}$. Let $\frac{u_M}{u_m}=C$, a constant.  Note that the number of violated pairs is at least $\frac{\varepsilon}{4C}\cdot |P_f|$.
 Therefore we have the following lemma.
\begin{lemma}\label{lem:prob}
If $f$ is $\eps$-far from $\cP$, then $f$ picks up a violated pair with probability at least $\frac{\eps}{4C}$ where $C=\frac{u_M}{u_m}$ is a constant.
\end{lemma}

\subsection{Testers for the hypergrid} \label{sec:hyper-test}
The hypergrid testers are easy consequences of the dimension
reduction and the line testers. Let $\cL_{n,d}$ denote the set of axis parallel lines in $[n]^d$.
\begin{shaded}
{\bf Hypergrid-Tester ($f$)}
\begin{asparaenum}
\itemsep0em 
\item Choose a line $\ell$ $i$ u.a.r. from $\cL_{n,d}$.\label{item:1b}
\item Run \textbf{Line-Tester($\frest{f}{\ell}$)}.\label{item:1c}
\item Repeat step \ref{item:1b}. $O(\frac{d}{\varepsilon})$ times.
\end{asparaenum}
\end{shaded}

\begin{lemma}\label{lem:hyper-bst} Consider a function $f$ that is $\varepsilon$-far from $\cP$. The  probability of rejection of Hypergrid Tester is at least $2/3$.
\end{lemma}

\begin{proof} From \Lem{grid-to-line-gen} we have $ \sum_{\ell \in \cL(n,d)}[L_1(\frest{f}{\ell},\cP)] \geq L_1(f,\cP)/2$. This can also be stated as 
$$\Exp_{\ell\sim \cL_{n,d}}[dist(\frest{f}{\ell},\cP)]\ge \frac{dist(f,\cP)}{2d}$$ where distances are measured with respect to $[n]$ and $[n]^d$ domains respectively. \Lem{prob} implies that the probability of picking up a violated pair in step \ref{item:1c}. is at least 
$\frac{\Exp_{\ell\sim \cL_{n,d}}[dist(\frest{f}{\ell},\cP)]}{2}\ge \frac{\varepsilon}{4d}$. Therefore, repeating it $O(d/\varepsilon)$ times boosts the rejection probability to $2/3$.

\end{proof}


\bibliographystyle{alpha}
\bibliography{derivative-testing}

\newcommand{\etalchar}[1]{$^{#1}$}
\begin{thebibliography}{BCGSM12}

\bibitem[AC06]{AC04}
N.~Ailon and B.~Chazelle.
\newblock Information theory in property testing and monotonicity testing in
  higher dimension.
\newblock {\em Inform.\ and Comput.}, 204(11):1704--1717, 2006.

\bibitem[ACCL07]{ACCL07}
N.~Ailon, B.~Chazelle, S.~Comandur, and D.~Liu.
\newblock Estimating the distance to a monotone function.
\newblock {\em Random Structures Algorithms}, 31(3):371--383, 2007.

\bibitem[AJMR12]{AJMS12}
P.~Awasthi, M.~Jha, M.~Molinaro, and S.~Raskhodnikova.
\newblock Testing {L}ipschitz functions on hypergrid domains.
\newblock In {\em Proceedings, International Workshop on Randomization and
  Computation (RANDOM)}, 2012.

\bibitem[BBM12]{BBM11}
E.~Blais, J.~Brody, and K.~Matulef.
\newblock Property testing lower bounds via communication complexity.
\newblock {\em Comp.\ Complexity}, 21(2):311--358, 2012.

\bibitem[BCGSM12]{BCG+10}
J.~Bri\"{e}t, S.~Chakraborty, D.~Garc\'{i}a-Soriano, and A.~Matsliah.
\newblock Monotonicity testing and shortest-path routing on the cube.
\newblock {\em Combinatorica}, 32(1):35--53, 2012.

\bibitem[BFR{\etalchar{+}}13]{BFRSW13}
Tugkan Batu, Lance Fortnow, Ronitt Rubinfeld, Warren~D. Smith, and Patrick
  White.
\newblock Testing closeness of discrete distributions.
\newblock {\em J. ACM}, 60(1):4, 2013.

\bibitem[BFRV11]{BFRV11}
Arnab Bhattacharyya, Eldar Fischer, Ronitt Rubinfeld, and Paul Valiant.
\newblock Testing monotonicity of distributions over general partial orders.
\newblock In Bernard Chazelle, editor, {\em ICS}, pages 239--252. Tsinghua
  University Press, 2011.

\bibitem[BGJ{\etalchar{+}}09]{BGJRW09}
A.~Bhattacharyya, E.~Grigorescu, K.~Jung, S.~Raskhodnikova, and D.~Woodruff.
\newblock Transitive-closure spanners.
\newblock In {\em Proceedings, ACM-SIAM Symposium on Discrete Algorithms
  (SODA)}, pages 531--540, 2009.

\bibitem[BRW05]{BRW05}
T.~Batu, R.~Rubinfeld, and P.~White.
\newblock Fast approximate {$PCP$}s for multidimensional bin-packing problems.
\newblock {\em Inform.\ and Comput.}, 196(1):42--56, 2005.

\bibitem[BRY14a]{BeRaYa14}
P.~Berman, S.~Raskhodnikova, and G.~Yaroslavtsev.
\newblock Testing with respect to $l_p$ distances.
\newblock In {\em Proceedings, ACM Symp. on Theory of Computing (STOC)}, 2014.

\bibitem[BRY14b]{BlJh+13}
E.~Blais, S.~Raskhodnikova, and G.~Yaroslavtsev.
\newblock Lower bounds for testing properties of functions on hypergrid
  domains.
\newblock In {\em Proceedings, IEEE Conference on Computational Complexity
  (CCC)}, March 2014.

\bibitem[CDJS14]{paper}
D.~Chakrabarty, K.~Dixit, M.~Jha, and C.~Seshadhri.
\newblock Property testing on product distributions: Optimal testers for
  bounded derivative properties.
\newblock \url{http://www.cs.princeton.edu/~csesha/product-full.pdf}, 2014.

\bibitem[CS13a]{ChSe13}
D.~Chakrabarty and C.~Seshadhri.
\newblock Optimal bounds for monotonicity and {L}ipschitz testing over
  hypercubes and hypergrids.
\newblock In {\em Proceedings, ACM Symp. on Theory of Computing (STOC)}, 2013.

\bibitem[CS13b]{ChSe13-2}
D.~Chakrabarty and C.~Seshadhri.
\newblock An optimal lower bound for monotonicity testing over hypergrids.
\newblock In {\em Proceedings, International Workshop on Randomization and
  Computation (RANDOM)}, 2013.

\bibitem[DDS{\etalchar{+}}13]{DDSVV13}
Constantinos Daskalakis, Ilias Diakonikolas, Rocco~A. Servedio, Gregory
  Valiant, and Paul Valiant.
\newblock Testing {\it k}-modal distributions: Optimal algorithms via
  reductions.
\newblock In Sanjeev Khanna, editor, {\em SODA}, pages 1833--1852. SIAM, 2013.

\bibitem[DGL{\etalchar{+}}99]{DGLRRS99}
Y.~Dodis, O.~Goldreich, E.~Lehman, S.~Raskhodnikova, D.~Ron, and
  A.~Samorodnitsky.
\newblock Improved testing algorithms for monotonicity.
\newblock In {\em Proceedings, International Workshop on Randomization and
  Computation (RANDOM)}, 1999.

\bibitem[DJRT13]{DiJh+13}
K.~Dixit, M.~Jha, S.~Raskhodnikova, and A.G. Thakurta.
\newblock Testing the {L}ipschitz property over product distributions with
  applications to data privacy.
\newblock In {\em Proceedings, Theory of Cryptography Conference (TCC)}, 2013.

\bibitem[DR11]{DolevR11}
E.~Dolev and D.~Ron.
\newblock Distribution-free testing for monomials with a sublinear number of
  queries.
\newblock {\em Theory of Computing}, 7(1):155--176, 2011.

\bibitem[EKK{\etalchar{+}}00]{EKK+00}
F.~Ergun, S.~Kannan, R.~Kumar, R.~Rubinfeld, and M.~Viswanathan.
\newblock Spot-checkers.
\newblock {\em J.\ Comput.\ System Sci.}, 60(3):717--751, 2000.

\bibitem[Fis04]{E04}
E.~Fischer.
\newblock On the strength of comparisons in property testing.
\newblock {\em Inform.\ and Comput.}, 189(1):107--116, 2004.

\bibitem[FLN{\etalchar{+}}02]{FLNRRS02}
E.~Fischer, E.~Lehman, I.~Newman, S.~Raskhodnikova, R.~Rubinfeld, and
  A.~Samorodnitsky.
\newblock Monotonicity testing over general poset domains.
\newblock In {\em Proceedings, ACM Symp. on Theory of Computing (STOC)}, 2002.

\bibitem[FV13]{FV13}
Vitaly Feldman and Jan Vondr{\'a}k.
\newblock Optimal bounds on approximation of submodular and xos functions by
  juntas.
\newblock In {\em FOCS}, pages 227--236. IEEE Computer Society, 2013.

\bibitem[GGL{\etalchar{+}}00]{GGLRS00}
O.~Goldreich, S.~Goldwasser, E.~Lehman, D.~Ron, and A.~Samorodnitsky.
\newblock Testing monotonicity.
\newblock {\em Combinatorica}, 20:301--337, 2000.

\bibitem[GGR98]{GGR98}
O.~Goldreich, S.~Goldwasser, and D.~Ron.
\newblock Property testing and its connection to learning and approximation.
\newblock {\em J.\ ACM}, 45(4):653--750, 1998.

\bibitem[GS09]{GlasnerS09}
D.~Glasner and R.~A. Servedio.
\newblock Distribution-free testing lower bound for basic boolean functions.
\newblock {\em Theory of Computing}, 5(1):191--216, 2009.

\bibitem[Hal06]{Hal-thesis}
S.~Halevy.
\newblock {\em Topics in Property Testing}.
\newblock PhD thesis, Tel Aviv University, 2006.

\bibitem[HK05]{HalevyK05}
S.~Halevy and E.~Kushilevitz.
\newblock A lower bound for distribution-free monotonicity testing.
\newblock In {\em Proceedings, International Workshop on Randomization and
  Computation (RANDOM)}, pages 330--341, 2005.

\bibitem[HK07]{HalevyK07}
S.~Halevy and E.~Kushilevitz.
\newblock Distribution-free property-testing.
\newblock {\em SIAM J.\ Comput.}, 37(4):1107--1138, 2007.

\bibitem[HK08a]{HalevyK08}
S.~Halevy and E.~Kushilevitz.
\newblock Distribution-free connectivity testing for sparse graphs.
\newblock {\em Algorithmica}, 51(1):24--48, 2008.

\bibitem[HK08b]{HK04}
S.~Halevy and E.~Kushilevitz.
\newblock Testing monotonicity over graph products.
\newblock {\em Random Structures Algorithms}, 33(1):44--67, 2008.

\bibitem[JR11]{JR11}
M.~Jha and S.~Raskhodnikova.
\newblock Testing and reconstruction of {L}ipschitz functions with applications
  to data privacy.
\newblock In {\em Proceedings, IEEE Symposium on Foundations of Computer
  Science (FOCS)}, 2011.

\bibitem[LR01]{LR01}
E.~Lehman and D.~Ron.
\newblock On disjoint chains of subsets.
\newblock {\em J.\ Combin.\ Theory Ser.\ A}, 94(2):399--404, 2001.

\bibitem[oCDVV14]{CDVV14}
Siu on~Chan, Ilias Diakonikolas, Paul Valiant, and Gregory Valiant.
\newblock Optimal algorithms for testing closeness of discrete distributions.
\newblock In Chandra Chekuri, editor, {\em SODA}, pages 1193--1203. SIAM, 2014.

\bibitem[PRR06]{PRR04}
M.~Parnas, D.~Ron, and R.~Rubinfeld.
\newblock Tolerant property testing and distance approximation.
\newblock {\em J.\ Comput.\ System Sci.}, 6(72):1012--1042, 2006.

\bibitem[RS96]{RS96}
R.~Rubinfeld and M.~Sudan.
\newblock Robust characterization of polynomials with applications to program
  testing.
\newblock {\em SIAM J.\ Comput.}, 25:647--668, 1996.

\bibitem[Val11]{Valiant11}
P.~Valiant.
\newblock Testing symmetric properties of distributions.
\newblock {\em SIAM J.\ Comput.}, 40(6):1927--1968, 2011.

\bibitem[VV11]{VV11}
Gregory Valiant and Paul Valiant.
\newblock The power of linear estimators.
\newblock In Rafail Ostrovsky, editor, {\em FOCS}, pages 403--412. IEEE, 2011.

\end{thebibliography}

\end{document}